\documentclass[11pt]{article}
\usepackage{amsmath}
\usepackage{amsfonts}
\IfFileExists{pdfsync.sty}{\usepackage{pdfsync}}{}
\usepackage[cyr]{aeguill}
\usepackage{multirow} 
\usepackage{versions}
  \includeversion{ignoreabstract}
  \includeversion{ignoreintro}
\includeversion{related}
  \includeversion{gametheory}
 \excludeversion{extensions}

\newcommand\TODOBY[2]{\centerline{\textbf{ #1 A FAIRE PAR #2 \\}}}

\newcommand\vectorl[1]{{\mathbf#1}}

\newcommand\vx{\vectorl{x}}
\newcommand\vy{\vectorl{y}}

\newcommand\N{\mathbb{N}}
\newcommand\Z{\mathbb{Z}}
\newcommand\II{{R}}
\newcommand\I{{I}}

\newtheorem{definition}{Definition}

\newtheorem{theorem}{Theorem}

\newtheorem{proposition}{Proposition}
\newenvironment{proof}{{\bf Proof}:}{\hfill $\Box$ }

 \newcommand\motnouv[1]{\emph{#1}}

\begin{document}

\title{Asymetric Pavlovian Populations\thanks{This work and all
    authors were partly supported by ANR Project SHAMAN, Xavier
    Koegler was supported by the ANR projects ALADDIN and PROSE}}

\author{Olivier Bournez$^1$, J\'er\'emie Chalopin$^2$, Johanne Cohen$^3$,\\   Xavier Koegler$^4$, Mika\"el Rabie$^5$ \\
{\small (1) Ecole Polytechnique \& Laboratoire d'Informatique (LIX), }\\  {\small 91128 Palaiseau Cedex, France. }\\
{\small email:Olivier.Bournez@lix.polytechnique.fr }\\
{\small(2) Laboratoire d'Informatique Fondamentale de Marseille, }\\ 
{\small CNRS \& Aix-Marseille Universit{\'e}, }\\ 
{\small  39 rue Joliot Curie, 13453 Marseille Cedex 13, France.}\\
{\small email:Jeremie.Chalopin@lif.univ-mrs.fr}\\
{\small (3) {PRiSM \& CNRS}, Univiersité de Versailles} \\ 
{\small 45 Avenue des Etats Unis, 78000 Versailles, France. }\\
{\small email :Johanne.Cohen@prism.uvsq.fr }\\
{\small (4) LIAFA \& Universit\'e Paris Diderot - Paris 7,}\\   
{\small 75205 Paris Cedex  13, France.} \\ 
{\small email :Xavier.Koegler@liafa.jussieu.fr }\\
{\small (5) ENS de Lyon \& Laboratoire d'Informatique (LIX),} \\ 
{\small 91128 Palaiseau
Cedex, France. }\\ 
{\small
email:Mikael.Rabie@lix.polytechnique.fr}}

\date{\today}

\maketitle

\begin{abstract}
\begin{ignoreabstract}
  Population protocols have been introduced by Angluin et {al.} as a
  model of networks consisting of very limited mobile agents that
  interact in pairs but with no control over their own movement. A
  collection of anonymous agents, modeled by finite automata, interact
  pairwise according to some rules that update their
  states.  
 Predicates on the initial configurations that can be computed by such
 protocols have been characterized as semi-linear predicates.  

   In an orthogonal way, several distributed systems have been termed
   in literature as being realizations of games in the sense of game
   theory.

  We investigate under which conditions population protocols, or more
  generally pairwise interaction rules, correspond to games.

  We show  that restricting to asymetric games is not
  really a
  restriction: all  predicates computable by
  protocols can
  actually be computed by protocols corresponding to games, i.e. any
  semi-linear predicate can be computed by a Pavlovian population
  multi-protocol.


\end{ignoreabstract}
\end{abstract}

\vspace{-1cm}


\section{Introduction}
\vspace{-1ex}

\begin{ignoreintro}

   The computational power of networks of anonymous resource-limited
  mobile agents has been investigated recently.  
  Angluin et al.\ proposed in \cite{AspnesADFP2004} the model of
  \motnouv{population protocols} where finitely many finite-state
  agents interact in pairs chosen by an adversary. Each interaction
  has the effect of updating the state of the two agents according to
  a joint transition function.
  A protocol is said to \motnouv{(stably) compute} a predicate on the
  initial states of the agents if, in any fair execution, after
  finitely many interactions, all agents reach a common output that
  corresponds to the value of the predicate.

The model has been originally proposed to model computations realized by
sensor networks in which passive agents are carried along by other
entities. Variants of the original model considered so far include restriction
to one-way communications~\cite{angluin2007cpp}, restriction to
particular interaction graphs~\cite{AngluinACFJP2005}, random
interactions~\cite{AspnesADFP2004}, with ``speed'' \cite{BeauquierB10}. Various kinds of fault tolerance
have been considered for population
protocols~\cite{Delporte-GalletFGR06}, including the search for
self-stabilizing solutions~\cite{AngluinAFJ2005}. 
Solutions to classical problems of distributed algorithms
have also been considered in this model (see \cite{Newmodels}).


Most of the works so far on population protocols have concentrated on
characterizing which predicates on the initial states can be computed
in different variants of the model and under various assumptions
\cite{Newmodels}. In particular, the predicates computable by the
unrestricted population protocols from \cite{AspnesADFP2004} have been
characterized as being precisely the semi-linear predicates, that is
those predicates on counts of input agents definable in first-order
Presburger
arithmetic~\cite{AspnesADFP2004,AngluinAE2006semilinear}. 
%

 
In an orthogonal way, pairwise interactions between finite-state agents are sometimes motivated by the study of the dynamics of particular two-player games from game theory. For example, the work in \cite{Ref9deFMP04} considers the dynamics of the so-called \textit{PAVLOV} behavior in the iterated Prisoners' Dilemma. Several results about the time of convergence of this particular dynamics towards the stable state can be found in \cite{Ref9deFMP04}, and \cite{FMP06}, for rings, and complete graphs~\cite{ICS11} with having various classes of adversarial schedulers~\cite{abs-0812-1194}.

Our purpose is to better understand whether and when pairwise
interactions, and hence population protocols, can be considered as the
result of a game. We prove the   result that restricting to
games is not really a restriction: all predicates computable by
protocols can actually be computed by protocols corresponding to
games, i.e. any semi-linear predicate can be computed by a Pavlovian
population multi-protocol.

%

 
In Section \ref{section:pp}, we    recall population
protocols. In Section \ref{section:gametheory}, we give some basics
from game theory. In Section \ref{sec:gamepp}, we discuss how a game
can be turned into a dynamics, and introduce the notion of {Pavlovian}
population. In Section \ref{sec:multiprotocols} we state our main
result: any
  semi-linear predicate can be computed by a Pavlovian population
  multi-protocol. Remaining sections correspond to its proof: we prove
  that threshold and modulo predicates can be computed respectively in
  Sections \ref{sec:counting} and \ref{sec:modulo}.

\end{ignoreintro}

\begin{related}
\medskip
\noindent\textbf{Related Works.}
As we already said, population protocols have been introduced in
\cite{AspnesADFP2004}, and proved to compute all semi-linear
predicates. They have been proved not to be able to compute more in
\cite{AngluinAE2006semilinear}. Various restrictions on the initial
model have been considered up to now.  An  survey
can be found in \cite{Newmodels}.


More generally, population protocols arise as soon as populations of
anonymous agents interact in pairs. Our original motivation was to
consider rules corresponding to two-player games, and population
protocols arose quite incidentally. The main advantage of the
\cite{AspnesADFP2004} settings is that it provides a clear
understanding of what is called a computation by the model. Many
distributed systems have been described as the result of games, but as
far as we know there has not been attempt to characterize what can be
computed by games in the spirit of this computational model.

In this paper, we turn two players games into dynamics over agents, by
considering \textit{PAVLOV} behavior. This is inspired by
\cite{Ref9deFMP04,FMP06,kraines1988psd} that consider the dynamics of
a particular set of rules termed the \textit{PAVLOV} behavior in the iterated
Prisoners' Dilemma. The \textit{PAVLOV} behavior is sometimes also termed
\textit{WIN-STAY, LOSE-SHIFT}
\cite{nowak1993sws,axelrod:1984:ec}. Notice, that we extended it from
two-strategies two-player games to $n$-strategies two-player games,
whereas above references only talk about two-strategies two-player
games, and mostly of the iterated Prisoners' Dilemma. This is clearly not
the only way to associate a dynamic to a game. Alternatives to \textit{PAVLOV} behavior could include \textit{MYOPIC} dynamics (at each step each player chooses the best response to previously played strategy by its adversary), or the well-known and studied \textit{FICTIOUS-PLAYER} dynamics (at each step each player chooses the best response to the statistics of the past history of strategies played by its adversary). We refer to \cite{theorylearninggames,LivreBinmore} for a presentation of results known about the properties of the obtained dynamics according to the properties of the underlying game. This is clearly non-exhaustive, and we refer to \cite{axelrod:1984:ec} for a zoology of possible behaviors for the particular iterated Prisoners' Dilemma game, with discussions of their compared merits.

Recently Jaggard et al. \cite{ICS11} studied a  distributed model
similar to protocol populations where the interactions between pairs
of agents
correspond to a game. Unlike in our model, each agent has there its own pay-off
matrix and has some knowledge of the history. This work gives several non-convergence results.

In this paper we consider possibly asymmetric games. In a
recent paper \cite{CSP08} we discussed population protocols
corresponding to Pavlovian strategies obtained from \emph{symmetric}
games and we gave some protocols to compute some basic
predicates. Unlike what we obtain here, where we prove that any
computable predicate is computable by a asymmetric Pavlovian population
protocol, restricting to symmetric games
seems a (too) strong restriction and most predicates (e.g. counting up
to $5$, to check where $x = 0 \mod 2$) seems not even
computable.
\end{related}

\vspace{-1ex}
\section{Population Protocols}
\label{section:pp}
\vspace{-1ex}

A protocol \cite{AspnesADFP2004} is given by $(Q,\Sigma,\iota,\omega,\delta)$ with the
following components. $Q$ is a finite set of \motnouv{states}.
$\Sigma$ is a finite set of \motnouv{input symbols}.  $\iota: \Sigma
\to Q$ is the initial state mapping, and $\omega: Q \to \{0,1\}$ is
the individual output function. $\delta \subseteq Q^4$ is a joint
transition relation that describes how pairs of agents can
interact. Relation $\delta$ is sometimes described by listing all
possible interactions using the notation $(q_1,q_2) \to (q'_1,q'_2)$,
or even the notation $q_1q_2 \to q'_1 q'_2$, 
for $(q_1,q_2,q'_1,q'_2) \in \delta$ (with the convention that
$(q_1,q_2) \to (q_1,q_2)$ when no rule is specified with $(q_1,q_2)$
in the left-hand side). The protocol is termed \motnouv{deterministic}
if for all pairs $(q_1,q_2)$ there is only one pair $(q'_1,q'_2)$ with
$(q_1,q_2) \to (q'_1,q'_2)$. In that case, we write
$\delta_1(q_1,q_2)$ for the unique $q'_1$ and $\delta_2(q_1,q_2)$ for
the unique $q'_2$.


Computations of a protocol proceed in the following way. The
computation takes place among $n$ \motnouv{agents}, where $n \ge 2$. A
\motnouv{configuration} of the system can be described by a vector of
all the agents' states. The state of each agent is an element of $Q$. Because agents
with the same states are indistinguishable, each configuration can be
summarized as an unordered multiset of states, and hence of elements
of $Q$.
Each agent is given initially some input value from $\Sigma$: Each
agent's initial state is determined by applying $\iota$ to its input
value. This determines the initial configuration of the population.

An execution of a protocol proceeds from the initial configuration by
interactions between pairs of agents. Suppose that two agents in state
$q_1$ and $q_2$ meet and have an interaction. They can change into
state $q'_1$ and $q'_2$ if $(q_1,q_2,q'_1,q'_2)$ is in the transition
relation $\delta$.  If $C$ and $C'$ are two configurations, we write
$C \to C'$ if $C'$ can be obtained from $C$ by a single interaction of
two agents: this means that $C$ contains two states $q_1$ and $q_2$
and $C'$ is obtained by replacing $q_1$ and $q_2$ by $q'_1$ and $q'_2$
in $C$, where $(q_1,q_2,q'_1,q'_2) \in \delta$. An \motnouv{execution}
of the protocol is an infinite sequence of configurations
$C_0,C_1,C_2,\cdots$, where $C_0$ is an initial configuration and $C_i
\to C_{i+1}$ for all $i\ge0$. An execution is \motnouv{fair} if for
every configuration $C$ that appears infinitely often in the execution,
if $C \to C'$ for some configuration $C'$, then $C'$ appears
infinitely often in the execution. As proved
in~\cite{AngluinAE2006semilinear}, the fairness condition implies that
any global configuration that is infinitely often reachable is
eventually reached. 

At any point during an execution, each agent's state determines its
output at that time. If the agent is in state $q$, its output value is
$\omega(q)$. The configuration output is $0$ (resp. $1$) if all
the individual outputs are $0$ (resp. $1$). If the individual
outputs are mixed $0$s and $1s$ then the output of the configuration
is undefined. 

Let $p$ be a predicate over multisets of elements of
$\Sigma$. Predicate $p$ can be considered as a function whose   range is
$\{0,1\}$ and whose domain is the collection of these multisets. The predicate is said to be computed by the protocol if,  for every  multiset $\I$, and
every fair execution that starts from the initial configuration
corresponding to $\I$, the output value of every agent eventually
stabilizes to $p(\I)$. Predicates can also be considered as  functions whose range is $\{0,1\}$ and whose domain is $\mathbb{N}^{|\Sigma|}$. 
The following is then known.

\vspace{-0.1cm}

\begin{theorem}[\cite{AspnesADFP2004,AngluinAE2006semilinear}] A
  predicate is computable in the population protocol model if and only
  if it is semilinear.
\end{theorem}
\vspace{-0.1cm}

Recall that semilinear sets are exactly the sets that are definable in
first-order Presburger arithmetic \cite{presburger:uvk}.

\vspace{-1ex}
\section{Game Theory}
\label{section:gametheory}
\vspace{-1ex}

\begin{gametheory}

We now recall the simplest concepts from Game Theory. We focus on non-cooperative games, with complete information, in normal form. 

The simplest game is made up of two players, called $\I$ (or
\emph{initiator}) and $\II$ (or \emph{responder}), with a finite set
of actions, called \emph{pure strategies}, $Strat(\I)$ and
$Strat(\II)$. Denote by $A_{i,j}$ (resp. $B_{i,j}$) the score
for player $\I$
(resp. $\II$) when $\I$ uses strategy $i \in Strat(\I)$ and $\II$ uses strategy
$j \in Strat(\II)$.
The scores are given by $n \times m$ matrices $A$ and $B$, where $n$ and
$m$ are the cardinality of $Strat(\I)$ and $Strat(\II)$.

A strategy $x$ in   $Strat(\I)$ is said to be a best response to strategy
$y$ in $Strat(\II)$, denoted by $x \in BR_A(y)$ if
 $A_{z,y} \le A_{x,y}$
 for all strategies $z \in Strat(\I)$. Conversely, a strategy $y \in
 Strat(\II)$ satisfies 
$y \in BR_B(x)$ if
 $B_{x,z} \le B_{x,y}$
 for all strategies $z \in Strat(\II)$. 
%
%
A pair $(x,y)$ is a \emph{(pure) Nash equilibrium} if $x \in
BR_A(y)$ and $y \in BR_B(x)$. %
%
%
In other words, two strategies $(x,y)$ form a Nash equilibrium if
in that state neither of the players has a unilateral interest to deviate
from it.

There are two main approaches to discuss dynamics of games. The first
consists in repeating games \cite{LivreBinmore}. The second in using models from
evolutionary game theory. Refer to \cite{Evolutionary1,LivreWeibull}
for a presentation of this latter approach.

Repeating $k$ times a game, is equivalent to extending the space of
actions into $Strat(\I)^k$ and $Strat(\II)^k$: player $\I$ (respectively
$\II$) chooses his or her action $\vx(t) \in Strat(\I)$, (resp. $\vy(t) \in
Strat(\II)$) at time $t$ for $t=1,2,\cdots,k$. This is
equivalent to a two-player game with respectively $n^k$ and $m^k$
choices for players.
%
%

In practice, player $I$ (respectively $\II$) has to solve the following
problem at each time $t$: given the history of the game up to now,
that is to say $X_{t-1}=\vx(1),\cdots,\vx(t-1)$ and
$Y_{t-1}=\vy(1),\cdots,\vy(t-1)$ what should I (resp. \II) play at time $t$? In
other words, how to choose $\vx(t) \in Strat(\I)$?  (resp. $\vy(t) \in
Strat(\II)$?)

Is is natural to suppose that this is
given by some behavior rules: $\vx(t)=f(X_{t-1},Y_{t-1})$ and $\vy(t)=g(X_{t-1},Y_{t-1})$ for some particular functions $f$ and $g$.

The question of the best behavior rule to use in games, in particular for the Prisoners' Dilemma
gave birth to an important literature. In particular, after the book
\cite{axelrod:1984:ec}, that describes the results of tournaments of
behavior rules for the iterated Prisoners' Dilemma, and that argues that
there exists a best behavior rule called $TIT-FOR-TAT$.
This  consists in cooperating at the first step,
and then do the same thing as the adversary at subsequent times. %
%
A lot of other behaviors, most of them with very
picturesque names have been proposed and studied: see for example
\cite{axelrod:1984:ec}.

Among possible behaviors there is \textit{PAVLOV} behavior: in the iterated
Prisoners' Dilemma, a player cooperates if and only if both players opted
for the same alternative in the previous move. This name 
\cite{axelrod:1984:ec,kraines1988psd,nowak1993sws} stems from the fact that this strategy embodies
an almost reflex-like response to the payoff: it repeats its former
move if it was rewarded  above a threshold value, but switches behavior if
it was punished by receiving under this value. Refer to
\cite{nowak1993sws} for some study of this strategy in the spirit of
Axelrod's tournaments.
The \textit{PAVLOV} behavior can also be termed \textit{WIN-STAY, LOSE-SHIFT} since
if the play on the previous round results in a success, then the
agent plays the same strategy on the next round. Alternatively, if the
play resulted in a failure the agent switches to another action
\cite{axelrod:1984:ec,nowak1993sws}.

\end{gametheory}

\vspace{-1ex}
\section{From Games To Population Protocols}
\label{sec:gamepp}
\vspace{-1ex}

In the spirit of the previous discussion, to any 
game, we can
associate a population protocol as follows, corresponding to 
a \textit{PAVLOV}(ian) behaviour:

\begin{definition}[Associating a Protocol to a Game] \label{def:1}
  Assume a (possibly asymmetric) two-player game is given. Let $A$ and
  $B$ be the corresponding matrices. Let $\Delta$ be
  some threshold. 
  
  The  protocol associated to the game  is a population
  protocol whose   set of states is $Q$, where  $Q=Strat(\I)=Strat(\II)$
  is the set of strategies of the game, and whose transition rules
  $\delta$ are
  given as follows: $(q_1,q_2,q'_1,q'_2) \in \delta$ where

\hspace{-0.8cm}
\begin{tabular}[h]{ccc}
  \begin{minipage}[h]{0.5\linewidth}
     \begin{itemize}
  \item $q'_1=q_1$ when $A_{q_1,q_2}\ge\Delta$,
  \item $q'_1 \in BR_A(q_2)$ when $A_{q_1,q_2} < \Delta$,
  \end{itemize}
  \end{minipage} &  &
\begin{minipage}[h]{0.5\linewidth}
     \begin{itemize}
  \item $q'_2=q_2$ when $B_{q_2,q_1} \ge \Delta$,
  \item $q'_2 \in BR_B(q_1)$ when $B_{q_2,q_1} < \Delta$.
  \end{itemize}
  \end{minipage}
\end{tabular}

\end{definition}
\vspace{-0.5cm}
\begin{definition}[Pavlovian Population Protocol]
A population protocol is \motnouv{Pavlovian} if it can be obtained
from a game as above
.
\end{definition}



A population protocol obtained from a game as above  will be termed \motnouv{deterministic}
if best responses are assumed to be unique; in this case, the
rules are deterministic: for all 
$q_1,q_2$, there is a unique $q'_1$ and a unique $q'_2$ such that
$(q_1,q_2,q'_1,q'_2) \in \delta$.



In order to avoid to talk about matrices, we start by stating some
structural properties of Pavlovian population protocols. 

\begin{proposition}\label{prop-sets} 
Consider a set of rules.  For all rules $ab \to a'b'$, we denote 
$\delta^I_a(b)=b'$ and $\delta^R_b(a)=a'$.
Let $Stable ^I(a) =\{x \in Q  | \delta^I_a(x)=x\},$
and 
 $Stable ^R(a) =\{x \in Q  | \delta^R_a(x)=x\}.$

Then the set of rules is deterministic Pavlovian iff $\forall a \in Q~ \exists
~max^I(a) \in Stable^I(a)$ and $\exists~ max^R(a) \in Stable^R(a)$ such that for all states $a$, 
\begin{enumerate}
\item $\forall b \not\in Stable^I(a)$ implies
  $\delta^I _a(b)=max^I (a)$. 
\item $\forall b \not\in Stable^R(a)$ implies
  $\delta^R _a(b)=max^R(a)$. 
\end{enumerate}
\end{proposition}

\begin{proof} First, we consider a Pavlovian population protocol $P$
 obtained from corresponding matrices $A$ and $B$. Let $\Delta$ be the
 associated threshold.  Let $a$ be an arbitrary state in $Q$, and let
 $q$ be the best response to strategy $a$ for matrix $B$.

 Focus on the rule $aq \to a'q'$ where $(a',q') \in Q^2$, i.e., focus on the case where player $I$ plays $a$ while player $\II$ plays
 $q$. As $q=BR_B(a)$, we have, by Definition~\ref{def:1}, 
 $q'$ equals to $q$. Thus, $q \in Stable^I(a)$.

Now, let consider $b$ such that $b \notin Stable^I(a)$.  We focus on
the rule $ab \to a''b'$ where $(a'',b')\in Q^2$. So by definition of
set $Stable^I$, we have $b \neq b'$. Using Definition~\ref{def:1}, we
have $B_{b,a} < \Delta$ and $b'=BR_B(a)$. So $b'=BR_B(a)=q$. Thus, if
we let $max^I(a) = q$, $max^I(a)$ satisfies the conditions of the
proposition.

 Using similar arguments, we can also prove that $ \exists ~max^R(a) \in
 Stable^R(a)$ such that $\forall b \not\in Stable^R(a)$ implies
 $\delta^I _a(b)=max^R (a)$.  In fact, we can sum up the relationship
 between the game matrix and rules by the following: for any   $a\in Q$, we have 
  $Stable^I(a) =\{x \in Q  |B_{x,a}\geq
 \Delta\} \cup \{BR_B(a)\}$ and  $max^I(a)=BR_B(a)$ and 
 $Stable^R(a) =\{x \in Q  |A_{x,a}\geq
 \Delta\} \cup \{BR_A(a)\}$  and $max^R(a)=BR_A(a)$.

Conversely, consider a population protocol $P$ satisfying the
properties of the proposition. 
All rules $ab \to a'b'$ are such that $\delta^I_a(b)=b'$ and
$\delta^R_b(a)=a'$.  We focus on the construction on a two-player game
having the corresponding matrices $A$, and $B$. We fix an arbitrary
value $\Delta$ as the threshold of the corresponding game.

\begin{itemize}
\item  If   $Stable^I(a) \neq Q$,  then   $B_{max^I(a),a} =\Delta+1$.  If  $x \in Stable^I(a)$ and if $x \neq max^I(a)$ then $B_{x,a} =\Delta$.  If  $x \notin Stable^I(a)$, then   $B_{x,a} =\Delta-1$. 
\item  If   $Stable^I(a) = Q$,  then   $\forall x\in Q$, $B_{x,a} =\Delta$.
\item If   $Stable^R(a) \neq Q$,  then   $A_{max^R(a),a} =\Delta+1$.  If  $x \in Stable^R(a)$ and if $x \neq max^R(a)$ then $A_{x,a} =\Delta$.  If  $x \notin Stable^R(a)$, then   $A_{x,a} =\Delta-1$. 
\item  If   $Stable^R(a) = Q$,  then   $\forall x\in Q$, $A_{x,a} =\Delta$.
\end{itemize}
It is easy to see that this game describes all rules of $P$. So,  $P$ is a Pavlovian population Protocol.
\end{proof}

\vspace{-1ex}
\section{Main Result}\label{sec:multiprotocols}
\vspace{-1ex}

Inverting value of the individual output function, the class of
predicates computable by a Pavlovian population protocol is clearly closed under
negation. 
However, this is not clear that predicates computable by Pavlovian
population protocols are closed under conjunction or disjunction.

This is true if one considers \emph{multi-protocol}. The idea is to
consider $k$ (possibly asymmetric) two-player games. At each step,
each player chooses a strategy for each of the $k$ games. Now each of
the $k$ games is played independently when two agents meet.
Formally:

\begin{definition}[Multiprotocol]
Consider $k$ (possibly asymmetric) two-player games. 
For game $i$,  let $Q^i$ be the corresponding
states, $A^i$ and $B^i$ the corresponding matrices.

The associated  population protocol is the population protocol whose   set of states is $Q=Q^1\times Q^2\times\ldots\times Q^k$, and whose  transition rules are given as follows:
$((q^1_1,\ldots,q^k_1),(q^1_2,\ldots,q^k_2),({q^1_1}',\ldots,{q^k_1}'),({q^1_2}',\ldots,{q^k_1}'))\in\delta$
where, for all $1 \le i\le k$, $(q^i_1,q^i_2,{q^i_1}',{q^i_2}')$ is a
transition of the Pavlovian population protocol associated to the $i_{th}$ game.

\end{definition}

Notice that, when considering population protocols, a multi-protocol
  is a particular population protocol. This is the key property used
  in \cite{AspnesADFP2004} to prove that stably computable predicates are closed under
  boolean operations.
When considering Pavlovian games, one can build multi-protocols that are not 
Pavlovian protocols, and it is not clear whether one can always transform any
pavlovian multi-protocol into an equivalent pavlovian protocol.  

As explained before, multisets of elements of $\Sigma = (\sigma_1,
\ldots, \sigma_l)$ are in  bijection with elements of
$\mathbb{N}^l$, and can be represented by a vector $(x_1, \ldots,
x_l)$ of non-negative integers where $x_i$ is the number of occurrences
of $\sigma_i$ in the multiset.  Thus, we consider predicates $\psi$ over
vectors of non-negative integers. We write $[\psi]$ for their
characteristic functions. 
Recall that a predicate is semi-linear iff it is Presburger
definable \cite{presburger:uvk}. Semi-linear predicate correspond to
boolean combinations of threshold predicates and modulo predicates
defined as follows (variables $x_i$ represent the number of agents
initially in state $\sigma_i$): 
%
%
  A \emph{threshold} predicate is of the form $[\Sigma a_ix_i \geq k]$,
  where $\forall i, a_i \in \Z$, $k \in \Z$ and the $x_i$s are
  variables.
  A \emph{modulo} predicate is of the form $[\Sigma a_ix_i \equiv b \mod k
]$, where $\forall i, a_i \in \Z$, $k \in \N\setminus\{0,1\}$, $b \in
  [1,k-1]$ and the $x_i$s are variables.

We can then state our main result:

\begin{theorem} \label{th:main}
For any predicate $\psi$, the following conditions are equivalent:
\vspace{-0.1cm}
\begin{itemize}
\item  $\psi$ is computable by a population protocol
\item  $\psi$ is computable by a Pavlovian population multi-protocol
\item  $\psi$ is semi-linear.
\end{itemize}
\end{theorem}
\vspace{-0.2cm}

The proof of the following proposition can be found in Appendix
 \ref{proof:closure}. 

\begin{proposition} \label{prop:closure}
The class of predicates computable by multi-games  are closed under boolean operations. 
\end{proposition}

As from Proposition \ref{prop:closure}, predicates
computable by Pavlovian  population multi-protocols are closed under
boolean operations, and as a Pavlovian population protocol is a
particular Pavlovian population multi-protocol, and as predicates
computable by (general)
population protocols are known to be exactly semi-linear predicates,
to prove 
Theorem \ref{th:main} we only need to prove that we can compute
{threshold} predicates and modulo predicates by Pavlovian
population protocols. This is the purpose of the following sections. 

\vspace{-1ex}
\section{Threshold Predicates}\label{sec:counting}
\vspace{-1ex}

In this section, we prove that we can compute threshold predicates
using Pavlovian protocols.

\begin{proposition}{\label{propsous}}
For any integer $k$, and any integers $a_1,a_2,\cdots, a_m$ there
exists a Pavlovian population protocol that computes $[\sum_{i=1}^m
  a_ix_i\ge k]$.
\end{proposition}

First note, that we can assume without loss of generality that $k \ge
1$. Indeed, $[\Sigma a_ix_i\ge -k]= [\Sigma (-a_i)x_i\le k] = [\Sigma
  (-a_i)x_i< k+1]$ which is the negation of $[\Sigma (-a_i) x_i\ge
  k+1]$. Thus from a population protocol computing $[\Sigma (-a_i)
  x_i\ge k+1]$ with $k \geq 0$, we just have to inverse the output
function to obtain a population protocol that computes $[\Sigma
  a_ix_i\ge -k]$. 

The purpose of the rest of this section is to
prove Proposition \ref{propsous}. We first discuss some basic ideas: 
Our techniques are inspired by the work of Angluin et
al.~\cite{angluin2007cpp}. The set of states we use is the
set of integers from $[-M,M]$ where $M = \max (|a_i|,2k-1)$. 
Each agent with input $\sigma_i$ is
given an initial weight of $a_i$. During the execution, the sum of the
weights over the whole population is
preserved. In~\cite{AngluinAE2006semilinear}, the general idea is the
following: two interacting agents with positive weights $p$ and $q$
such that $p+q \leq M$ are transformed into an agent with weight $0$
and an agent with weight $p+q$, while two agents with weight $p$ and
$q$ such that $p+q > M$ are transformed into two agents with weight 
$\lfloor (p+q)/2 \rfloor$ and $\lceil (p+q)/2 \rceil$ that are both
greater or equal to $k$.

In our setting, we cannot use the same rules since all agents that
change their states when they meet an agent in state $p$ while being
initiator (resp. responder) must take the same state that only depends
of $p$. To avoid this problem, a trick is to use rules of  the following form: $pq \rightarrow (p+1)(q-1)$. However, we also have to make
sure that the protocol enables all agents to agree in the final
configuration. Whereas this kind of consideration is easy in the classical
population protocol model, this turns out to be tricky in our
settings.

We describe a protocol that computes
 $[\sum_{i=1}^m a_ix_i\ge k]$. Our protocol is defined as follows:  we
 consider $\Sigma=\{\sigma_1, \ldots, \sigma_l\}$, $Q=\{\top\} \cup
 [-M,M]$; for all $i$, $\iota(\sigma_i)=a_i$; and we take $\omega
 (\top) = 1$ and for any $p \in [-M, M]$, $\omega(p) = 1$ if and only
 if $p \geq 1$.

We distinguish two cases: either $k = 1$, or $k \geq 2$. We present
two protocols here, because we need to have a mechanism in our
protocols to enable to ``broadcast'' the result; this is not so
difficult in the first case whereas it is more technical in the second
one.  
Due to lack of space, we only give the rules for $k=1$ (the proof can
be found in Appendix~\ref{proof:casun}), but provide a full proof for
the case $k \ge 2$.

\noindent\textbf{Case $\mathbf{k = 1}$.}
Our protocol computing $[\Sigma a_ix_i \geq 1]$ is defined as follows (see 
Appendix \ref{proof:casun}).
%
The rules are the
following. 

\noindent{\small{
\begin{tabular}{ccc}
$
\begin{array}{cc@{\rightarrow}cl}
 \hspace{1cm}  &\top\top  &\top \top  \\
&1\top   & 1\top \\
&\multicolumn{3}{l}{}\\
\multicolumn{4}{l}{\forall n \in [-M,0],  \forall p \in [2,M-1]}  \\
&n\top  & n 0   &  \\ 
&nx  & nx & \forall x \in [-M,M],  \\
\end{array}
$
&  \hspace{1cm}  &
$
\begin{array}{c@{\rightarrow}cl}
\multicolumn{3}{l}{} \\
\top x  & \top x & \forall x \in [-M,M] \\
1n  & (n+1)\top & \forall n \in [-M,0]  \\
1p & 1p & \forall p \in [1,M]\\
\multicolumn{3}{l}{} \\
 p\top  & p\top &  \\
pn & (n+1)(p-1) &  \\ 
pp' & pp' & \forall p' \in [1,M]\\
\end{array}
$\\
\end{tabular} 
}
}

\noindent\textbf{Case $\mathbf{k \ge 2}$.}
Our protocol is deterministic and from Proposition~\ref{prop-sets}
uniquely determined by the sets $Stable^I(q)$, $Stable^R(q)$, and by
the values $max^I(q)$, $max^R(q)$ defined as follows. 
{\small $$
\begin{array}{|c|c|c|c|c|}
\hline
 q \in Q & Stable^I(q) & max^I(q) & Stable^R(q) & max^R(q) \\ 
 \hline
\top & \{\top\} \cup [-M,0] \cup [k,M] & -1 & \{\top\} \cup [-M,M] &\\

n \in [-M,-1] & [-M,M] & 0 & \{\top\} \cup [-M,0] & (n+1)\\

0 & [-M,M] & 0 & \{\top\} \cup [-M,k-1] & 1\\

1 & \{\top,0,M\} & \top & [-M,0] & 2 \\

p \in [2,k-1] & \{\top,0,M\} & (p-1) & [-M,0] & (p+1) \\

b \in [k,M-1] & \{\top\} \cup [k,M] & (b-1) & \{\top\} \cup [-M,0]
\cup [k,M] & (b+1) \\

M & \{\top\} \cup [k,M] & (M-1) & \{\top\} \cup [-M,M] &\\
\hline
\end{array}
$$}

The transition rules we obtain from these sets and values are the 
following.

\noindent{\small{
\begin{tabular}{cc}
$
\begin{array}{cc@{\rightarrow}cl}
\hspace{1cm}&\top\top  &\top \top  \\
&\top p  & (p+1)(-1) & \forall p \in [1,k-1]  \\
& 1\top   & 1\top \\
& 1x  & (x+1)\top & \forall x \notin \{\top,0,M\}  \\
\multicolumn{4}{l}{\forall p \in [2,k-1]} \\
& p\top  & p\top &\\
& p0  & p0 &  \\
\multicolumn{4}{l}{\forall n \in [-M,0]} \\
& n\top  & n 0\\ 
\multicolumn{4}{l}{\forall b \in  [k,M]} \\
&b\top  & b\top &  \\
& bx  & (x+1)(b-1)& \forall x \in [-M,k-1] \\
\end{array}
$
& 
$
\begin{array}{c@{\rightarrow}ll}
\top x  & \top x & \forall x \in [-M,0] \cup [k,M] \\
 \multicolumn{3}{l}{}\\
10  & 10\\
 1M  & 1M \\
 \multicolumn{3}{l}{}\\
px & (x+1)(p-1) & \forall x \notin \{\top,0,M\}\\ 
 pM  & pM & \\
\multicolumn{3}{l}{} \\
nx  & nx & \forall x \in [-M,M]\\
\multicolumn{3}{l}{} \\
 bb' & bb' & \forall b' \in  [k,M]\\
\multicolumn{3}{l}{} \\

\end{array}
$\\
\\
\end{tabular} 
}
}

We say that an agent in state $x \in [-M,M]$ has weight $x$ and that
an agent in state $\top$ has weight $0$. Note that in the initial
configuration the sum of the weights of all agents is exactly $\Sigma
a_i x_i$. Note that any of the rule of our protocol does not modify
the total weight of the population, i.e., at any step of the
execution, the sum of the weights of all agents is exactly $\Sigma a_i
x_i$.

Note that the stable configurations, (i.e., the configurations where
no rule can be applied to modify the state of any agent), are the
following:
\vspace{-1ex}
\begin{itemize}
\item every agent $a$ is in some state $n(a) \in [-M,0]$,
\item a unique agent is in state $p \in [1,k-1]$ and every
  other agent is in state $0$. 
\item every agent $a$ is either in some state $b(a) \in [k,M]$ or in
  state $\top$.
\end{itemize}
\vspace{-1ex}

Note that no agent starts in state $\top$, and that no rule enables
the two interacting agents to enter the state $\top$ except for the rule
$\top\top \rightarrow \top\top$. Thus, we know that it is impossible
that all agents are in state $\top$. Consequently, in the last case
described, we know that there is at least one agent in a state $b\in
[k,M]$.

Note that in any stable configuration, all agents have the same
output; if $\Sigma a_ix_i \geq k$ then all agents output $1$, while in
all the other cases, the agents output $0$.  
Thus, if the population reaches a stable configuration, we know that the
computed output is correct and that it will not be modified any
more. Now, we should prove that the fairness condition ensures that we
always reach a stable configuration. In fact, it is sufficient to
prove that from any reachable configuration, there exists an execution
that reaches a stable configuration.


Consider any configuration reached during the execution. As long as
there is an agent in state $p \in [1,M]$ and an agent in state $n \in
[-M,-1]$, we apply $pn \rightarrow (n+1)(p-1)$.  Thus we can always
reach a configuration where the states of all agents are in
$[-M,0]\cup\{\top\}$ if $\Sigma a_ix_i \leq 0$, or in
$[0,M]\cup\{\top\}$ otherwise.

If $\Sigma a_ix_i \leq 0$, then there is at least one agent in state
$n \in [-M,0]$, since all agents cannot be in state $\top$. In this
case, applying iteratively the rule $n\top \rightarrow n0$, we reach a
stable configuration where all agents have a state in $[-M,0]$.

Suppose now that $\Sigma a_ix_i \in [1,k-1]$. Since $\Sigma a_ix_i \in
[1,k-1]$, each agent with a positive weight is in a state in
$[1,k-1]$. Applying iteratively the rule $pp' \rightarrow (p-1)(p'+1)$
where $p,p' \in [1,k-1]$, we reach a configuration where there is
exactly one agent in state $p \in [1,k-1]$ while all the other agents
are in state $0$ or $\top$. Applying iteratively the rules $\top p
\rightarrow (p+1)(-1)$ and $(p+1)(-1) \rightarrow 0p$, we reach a
configuration where one agent is in state $p \in [1,k-1]$ while all
the other agents are in state $0$.

Finally, assume that $\Sigma a_ix_i \geq k$. If there is an agent in
state $p \in [1,k-1]$, we know that there is at least another agent in
state $q \in [1,M]$. If $p + q \leq M$, applying iteratively the rule
$pq \rightarrow (p-1)(q+1)$ between these two agents, we reach a
configuration where one of these two agents is in state $0$ while the
other is in state $p+q$. In this case, we have strictly reduced the
number of agents in a state in $[1,k-1]$. If $p+q > M \geq 2k$, then
$q \in [k,M]$, and applying iteratively the rule $qp \rightarrow
(q-1)(p+1)$, we reach a configuration where one agent is in state $k$
while the other agent is in state $p+q-k \in [k,2M]$. Here again, we
have strictly reduced the number of agents in a state in
$[1,k-1]$. Applying these rules as long as there exists an agent in
state $p \in [1,k-1]$, we reach a configuration where all agents are
either in a state in $[k,M]$, or in state $0$ or $\top$. Since $\Sigma
a_ix_i \in [k,M]$, we know there exists an agent in state $b \in
[k,M]$. Applying iteratively the rules $b0 \rightarrow 1(b-1)$ and
$1(b-1) \rightarrow b\top$, we reach a stable configuration where all
agents are either in state $\top$ or in a state in $[k,M]$.

\vspace{-1ex}
\section{Modulo Counting}\label{sec:modulo}
\vspace{-1ex}

\begin{proposition}{\label{prop-modulo}}
For any integers $k$, $b$, and any integers $a_1,a_2,\cdots, a_m$ there exists a Pavlovian population protocol that
computes $[\sum_{i=1}^m a_ix_i \equiv b \mod k]$.
\end{proposition}

Due to lack of space, we only give the rules of the protocol for the
case when $b \in [1,k-1]$ (see Appendix \ref{proof:casmodulo} for the
complete proof). In that case, our protocol is defined as follows:
$\Sigma=\{\sigma_1, \ldots, \sigma_l\}$, $Q=\{\top\} \cup [0,k-1]$;
for all $i$, let $\iota(\sigma_i) \equiv a_i \mod k$; let $\omega
(\top) = 1$ and for any $p \in [0, k-1]$, let $\omega(p) = 1$ if and
only if $p = b$.

The rules are the following:

\noindent{\small{
\begin{tabular}{ccc}
\\
$
\begin{array}{cc@{\rightarrow}cl}
\hspace{1cm}& \top\top  &\top \top  \\
& b\top  & b\top & \\
& 0\top  & 0\top & \\
 \multicolumn{4}{l}{\forall p \in [1,k-1]}\\
& \top p  & \top p &  \\
& p(k-1) & \top (p-1)\\
 \multicolumn{4}{l}{\forall  p \in [0,k-1]\setminus\{b\}}\\
&0p  & 0p & \\
 \end{array}
$
& \hspace{1cm}&
$
\begin{array}{r@{\rightarrow}cl}
\top 0  & 0 0 & \\
0b  & \top(k-1)& \textrm{if } b =  k-1\\
0b  & (b+1)(k-1)& \textrm{if } b\neq k-1\\
 \multicolumn{3}{l}{}\\
 pp'  & pp' & \forall p' \in [0,p-1] \\
 pp'  & (p'+1)(p-1) & \forall p' \in [p,k-2] \\

 \multicolumn{3}{l}{}\\
p \top  & 1 (p-1) &\\
 \end{array}
$
\end{tabular}
}}

\begin{extensions}
\section{Extensions}\label{sec:conclusion}

\TODOBY{ ECRIRE CA}{ Olivier } 

\begin{theorem}
Pavlovian immediate observation protocols compute exactly the same predicates as
immediate observation protocols, that is to say $COUNT_*$.
\end{theorem}

IDEE PREUVE: c'est que leur algo est ok.

We don't know about Pavlovian immedate transmission.

REMARK:
\begin{itemize}
\item ATTENTION: PAVLOVIAN $\neq$ PAVLOVIAN AVEC CHGT AUTRE PAPIER
  (SUGGERE STRONGLY PAVLOVIAN DANS AUTRE)
\item IMMEDIATE OBSERVATION = 1 JEU EXTREMEMENT NATUREL.
\item IMMEDIATE TRANSMISSION= moins naturel. Stragégie ``aveugle. 
\end{itemize}
\end{extensions}

\vspace{-1ex}
\section{Conclusion}
\vspace{-1ex}

In this work, we present some (original an non-trivial) Pavlovian population protocols that compute
the general threshold and modulo predicates. From this, we deduced that
a predicate is computable in the Pavlovian population multi-protocol
model if and only if it is semilinear.  

In other words, we proved that restricting to rules that correspond to asymmetric
games in pairwise interactions is not a restriction. 

We however needed to consider multi-protocols, that is to say
multi-games. We conjecture that the Pavlovian population protocols
(i.e. non-multi-protocol) can not compute all semilinear predicates.
A point is that in such protocols the set of rules are very limited
(see Proposition~\ref{prop-sets}). In particular, it seems rather
impossible to perform an ``or'' operation between two modulo
predicates in the general case.

Notice that the hypothesis of asymmetric games seems also
necessary. We studied symmetric Pavlovian population protocols in
\cite{CSP08} where we demonstrated that some non-trivial predicates
can be computed. However, even very basic predicates,
like the threshold predicate counting up to $5$, seems problematic to be computed by symmetric
games. With asymmetric games, general threshold and modulo predicates
can be computed.

\bibliographystyle{plain}


\newpage
\appendix

\section{Proofs}

The following proofs are here for completeness of the refereeing
process. 

\subsection{Proof of Proposition 
\ref{prop:closure}}
\label{proof:closure}

The negation is easy to deal with, as we just need to change $\omega$
into $1-\omega$. From de Morgan's laws, we only need to prove closure
by conjunction.

For the conjunction of two multi-games, let consider the multi-game
including all the games present in the two initial multi-games. The
conjunction protocol is the one associated to the new multi-game with
output
function $$\omega'(q^1,\ldots,q^k,p^1,\ldots,p^l)=\omega_1(q^1,\ldots,q^k)*\omega_2(p^1,\ldots,p^l),$$
where $(q^1,\ldots,q^k)$ and $(p^1,\ldots,p^l)$ are the corresponding
games in the first and second multi-games, and $\omega_1$ and
$\omega_2$ the respective output functions.


\subsection{Proof of Our Protocol for Threshold Predicates when $k=1$}
\label{proof:casun}

Our protocol computing $[\Sigma a_ix_i \geq 1]$ is defined as follows:  
\begin{itemize}
\item $\Sigma=\{\sigma_1, \ldots, \sigma_l\}$.
\item $Q=\{\top\} \cup [-M,M]$, where $\top$ corresponds to a weight
  of $0$ but has a different output meaning than the state $0$. 
\item $\iota(\sigma_i)=a_i$.
\item $\omega (\top) = 1$ and for any $p \in [-M, M]$,
  $\omega(p) = 1$ if and only if $p \geq 1$. 
\end{itemize}

Our protocol is deterministic and from Proposition~\ref{prop-sets}
uniquely determined by the sets $Stable^I(q)$, $Stable^R(q)$, and by
the values $max^I(q)$, $max^R(q)$ defined as follows. 
$$
\begin{array}{|c|c|c|c|c|}
\hline
q \in Q & Stable^I(q) & max^I(q) & Stable^R(q) & max^R(q) \\ 
\hline
\top & \{\top\} \cup [-M,M] &   & \{\top\} \cup [-M,M] &\\

n \in [-M,-1] & [-M,M] & 0 & \{\top\} \cup [-M,0] & (n+1)\\

0 & [-M,M] & 0 & \{\top\} \cup [-M,1] & 1\\

1 & \{\top\}\cup [1,M] & \top &  \{\top\} \cup [-M,M] &  \\

p \in [2,M] & \{\top\}\cup [1,M] & (p-1) &  \{\top\} \cup [-M,M] & \\

\hline
\end{array}
$$

The transition rules we obtain from these sets and values are the
following. 
\noindent{{
\begin{tabular}{ccc}
$
\begin{array}{cc@{\rightarrow}cl}
 \hspace{1cm}  &\top\top  &\top \top  \\
&1\top   & 1\top \\
&\multicolumn{3}{l}{}\\
\multicolumn{4}{l}{\forall n \in [-M,0],  \forall p \in [2,M-1]}  \\
&n\top  & n 0   &  \\ 
&nx  & nx & \forall x \in [-M,M],  \\
\end{array}
$
&  \hspace{1cm}  &
$
\begin{array}{c@{\rightarrow}cl}
\multicolumn{3}{l}{} \\
\top x  & \top x & \forall x \in [-M,M] \\
1n  & (n+1)\top & \forall n \in [-M,0]  \\
1p & 1p & \forall p \in [1,M]\\
\multicolumn{3}{l}{} \\
 p\top  & p\top &  \\
pn & (n+1)(p-1) &  \\ 
pp' & pp' & \forall p' \in [1,M]\\
\end{array}
$\\
\end{tabular} 
}
}

We say that an agent in state $x \in [-M,M]$ has weight $x$ and that
an agent in state $\top$ has weight $0$. Note that in the initial
configuration the sum of the weights of all agents is exactly $\Sigma
a_i x_i$. Note that any of the rule of our protocol does not modify
the total weight of the population, i.e., at any step of the
execution, the sum of the weights of all agents is exactly $\Sigma a_i
x_i$.

Note that the stable configurations, (i.e., the configurations where
no rule can be applied to modify the state of some agent), are the
following:
\begin{itemize}
\item each agent $a$ is in some state $n(a) \in [-M,0]$,
\item each agent $a$ is either in some state $p(a) \in [1,M]$ or in
  state $\top$.
\end{itemize}

Note that no agent starts in state $\top$, and that no rule enables
the two interacting agents to enter the state $\top$ except for the rule
$\top\top \rightarrow \top\top$. Thus, we know that it is impossible
that all agents are in state $\top$. Consequently, in the last case
described, we know that there is at least one agent in a state $p\in
[1,M]$.

Note that in any stable configuration, all agents have the same
output; if $\Sigma a_ix_i \geq 1$ then all agents output $1$, while in
all the other cases, the agents output $0$. 

Thus, if the population reaches a stable configuration, we know that the
computed output is correct and that it will not be modified any
more. Now, we should prove that the fairness condition ensures that we
always reach a stable configuration. In fact, it is sufficient to
prove that from any reachable configuration, there exists an execution
that reaches a stable configuration. Indeed, since the number of
configurations is finite, it means that in any execution, there is a
stable configuration $S$ such that from any reachable configuration, 
$S$ is reachable. The fairness condition ensures that $S$ is
eventually reached. 

We now show that from any configuration, there exists an execution
that leads to a stable configuration.

Consider any configuration reached during the execution. As long as
there is an agent in state $p \in [1,M]$ and an agent in state $n \in
[-M,-1]$, we apply $pn \rightarrow (n+1)(p-1)$.  Thus we can always
reach a configuration where the states of all agents are in
$[-M,0]\cup\{\top\}$ if $\Sigma a_ix_i \leq 0$, or in
$[0,M]\cup\{\top\}$ otherwise.

If $\Sigma a_ix_i \leq 0$, then there is at least one agent in state
$n \in [-M,0]$, since all agents cannot be in state $\top$. In this
case, applying iteratively the rule $n\top \rightarrow n0$, we reach a
stable configuration where all agents have a state in $[-M,0]$.

If $\Sigma a_ix_i \geq 1$, then there is at least one agent in state
$p \in [1,M]$. If there is an agent in state $0$, we apply the
following rules.  If there is no agent in state $1$, we apply a rule
$p0 \rightarrow 1(p-1)$ for $p\in [2,M]$ to create a $1$. Then,
applying iteratively the rule $10 \rightarrow 1\top$, we reach a
stable configuration where each agent is either in a state in $[1,M]$,
or is in state $\top$.

This ends the proof of Proposition~\ref{propsous} when $k = 1$.






\subsection{Proof of Our Protocol for Modulo Counting}
\label{proof:casmodulo}

Here again, we use rules of the form $np \rightarrow (n+1)(p-1)\mod k$ so
that at the end of the computation there is only one node with a
non-zero weight. We should also add some rules to be able to ensure
the ``broadcast'' of the result among the agents. We consider two
  cases, depending on whether $b$ is equal to zero or not. 



\subsubsection{Case $b \in [1,k-1]$.}

Our protocol is defined as follows: 
\begin{itemize}
\item $\Sigma=\{\sigma_1, \ldots, \sigma_l\}$.
\item $Q=\{\top\} \cup [0,k-1]$, where $\top$ corresponds to a weight
  of $0$ but has a different output meaning than the state $0$. 
\item $\iota(\sigma_i) \equiv a_i \mod k$.
\item $\omega (\top) = 1$ and for any $p \in [0, k-1]$,
  $\omega(p) = 1$ if and only if $p = b$. 
\end{itemize}

Our protocol is deterministic and from Proposition~\ref{prop-sets}
uniquely determined by the sets $Stable^I(q)$, $Stable^R(q)$, and by
the values $max^I(q)$, $max^R(q)$ defined as follows. In the following
table, when $p = k-1$ (resp. $b = k-1$), $p+1$ (resp. $b+1$) should be
understand as $\top$.

$$
\begin{array}{|c|c|c|c|c|}
\hline
q \in Q & Stable^I(q) & max^I(q) & Stable^R(q) & max^R(q) \\ 
\hline
\top & Q &  & \{\top, 0, b\}  & 1 \\

0 & Q\setminus\{b\} & (k-1) & [0,k-1] & 0\\

b & \{\top\} \cup [0,b-1] & (b-1) & \{\top\}\cup[b+1,k-1] & (b+1) \\

p \in [1,k-1]\setminus\{b\} & [0,p-1] & (p-1) & \{\top,0\}\cup[p+1,k-1] & (p+1) \\

\hline
\end{array}
$$

The transition rules we obtain from these sets and values are the
following.

\noindent{{
\begin{tabular}{ccc}
\\
$
\begin{array}{cc@{\rightarrow}cl}
\hspace{1cm}& \top\top  &\top \top  \\
& b\top  & b\top & \\
& 0\top  & 0\top & \\
 \multicolumn{4}{l}{\forall p \in [1,k-1]}\\
& \top p  & \top p &  \\
& p(k-1) & \top (p-1)\\
 \multicolumn{4}{l}{\forall  p \in [0,k-1]\setminus\{b\}}\\
&0p  & 0p & \\
 \end{array}
$
& \hspace{1cm}&
$
\begin{array}{r@{\rightarrow}cl}
\top 0  & 0 0 & \\
0b  & \top(k-1)& \textrm{if } b =  k-1\\
0b  & (b+1)(k-1)& \textrm{if } b\neq k-1\\
 \multicolumn{3}{l}{}\\
 pp'  & pp' & \forall p' \in [0,p-1] \\
 pp'  & (p'+1)(p-1) & \forall p' \in [p,k-2] \\

 \multicolumn{3}{l}{}\\
p \top  & 1 (p-1) &\\
 \end{array}
$
\end{tabular}
}}

Note that in the initial configuration the sum of the weights of all
agents is exactly $\Sigma a_i x_i \mod k$. Note that the application of
any rule of our protocol does not modify this value, i.e., at any step
of the execution, the sum of the weights of all agents is exactly
$\Sigma a_i x_i \mod k$.

Note that the stable configurations, (i.e., the configurations where
no rule can be applied to modify the state of some agent), are the
following:
\begin{itemize}
\item there exists a unique agent in state $b$ and all
  other agents are in state $\top$. 
\item there exists at most one agent in state $p \in
  [1,k-1]\setminus\{b\}$ and all other agents are in state $0$. 
\end{itemize}

As in the protocols for threshold predicates, no rule can transform
the states of two interacting agents into $\top$ except for the rule
$\top\top \rightarrow \top\top$. Since initially, no agent is in state
$\top$, it is impossible to reach a configuration where all agents are
in state $\top$. 

In any stable configuration, either all agents output $1$ (if $\Sigma
a_ix_i \equiv b \mod k$), or all agents output $0$. We now show that from
any reachable configuration, we can reach a stable configuration.

As long as there are two agents in states $p,p' \in [1,k-1]$ with $p
\leq p'$, we can apply iteratively the rule  $pp' \rightarrow
(p'+1)(p-1)$ between these two agents. Doing so, either one agent
enters state $0$, or one agent is in state $k-1$, while the other is
in state $p'' \in [1,k-1]$. In the latter case, applying the rule
$p(k-1)\rightarrow \top(p-1)$, we also decrease the number of agents
with a positive weight by $1$. Thus, we can always reach a
configuration where there is at most one agent $a$ in state $p \in
[1,k-1]$, while all the other agents are in state $0$ or $\top$. 

If $\Sigma a_ix_i \equiv 0 \mod k$, then all agents are in state $0$ or
$\top$, and we know there is at least one agent in state $0$. Applying
as long as necessary the rule $\top 0 \rightarrow 00$, we reach a
stable configuration where all agents are in state $0$.

If $\Sigma a_ix_i \equiv p \mod k$ with $p \in [1,k-1] \setminus \{b\}$
then one agent is in state $p$ while all the other agents are in state
$0$ or $\top$. If there is an agent in state $\top$, we apply the rule
$p\top \rightarrow 1(p-1)$. If $p =1$, then there is one more agent in
state $0$. If $p > 1$, we apply the rule $1(p-1)\rightarrow p0$, to
also get one more agent in state $0$. Iterating this process, we reach
a stable configuration where one agent is in state $p$ while all other
agents are in state $0$.

If $\Sigma a_ix_i \equiv b \mod k$, then one agent is in state $b$ while
all the other agents are in state $0$ or $\top$. As long as there is
an agent in state $0$, if $b \neq k-1$, we apply the rules $0b
\rightarrow (b+1)(k-1)$ and $(b+1)(k-1) \rightarrow \top b$.  If $b =
k-1$, we apply the rule $0(k-1)\rightarrow \top(k-1)$. Doing so, we
reach a stable configuration where one agent is in state $b$ while all
other agents are in state $\top$.

\subsubsection{Case $b=0$.}

If $k = 2$, then $[\Sigma a_ix_i \equiv 0 [2]]$ is the negation of
$[\Sigma a_ix_i \equiv 1 [2]]$. In this case, since pavlovian protocols
are closed under negation, we know from the previous case that there is a
pavlovian protocol that computes $[\Sigma a_ix_i \equiv 0 [2]]$. 

In the following, we assume that $k \geq 3$. Our protocol is defined
as follows:  
\begin{itemize}
\item $\Sigma=\{\sigma_1, \ldots, \sigma_l\}$.
\item $Q=\{A,B\} \cup [0,k-1]$, where $A$ and $B$ corresponds to a
  weight of $0$ but has a different output meaning than the state $0$.
\item $\iota(\sigma_i) \equiv a_i \mod k$.
\item $\omega (0) = 1$ and for any $x \in [1, k-1]\cup\{A,B\}$,
  $\omega(x) = 0$.
\end{itemize}

Our protocol is deterministic and from Proposition~\ref{prop-sets}
uniquely determined by the sets $Stable^I(q)$, $Stable^R(q)$, and by
the values $max^I(q)$, $max^R(q)$ defined as follows.

$$
\begin{array}{|c|c|c|c|c|}
\hline
&&&&\\
q \in Q & Stable^I(q) & max^I(q) & Stable^R(q) & max^R(q) \\ 
&&&&\\
\hline
A & Q\setminus\{B\} & 0 & Q\setminus\{B\} & 0 \\

B & Q\setminus\{A\} & 0 & Q\setminus\{A\} & 0 \\

0 & \{A,B,0,1,(k-1)\} & (k-1) & [0,k-1] & 0\\

1 & \{A\} & A & Q\setminus\{1\} & 2 \\

(k-1) & Q\setminus\{(k-1)\} & (k-2) & \{B\} & B \\

p \in [2,k-2]& \{A,B\}  \cup [0,p-1] & (p-1) & \{A,B\}\cup[p+1,k-1] &
(p+1) \\ 

\hline
\end{array}
$$

The transition rules we obtain from these sets and values are the
following. 
\noindent{\small{
\begin{tabular}{cc}
\\
$
\begin{array}{c@{\rightarrow}cl}
AA  &AA  \\
AB  & 00 & \\
A0  & 00 & \\
Ap  & Ap & \forall p \in [1,k-2]\\
A(k-1)  & B(k-1) \\
\multicolumn{3}{l}{}\\
BA  & 00  \\
BB  & BB & \\
B0  & 00 & \\
Bp  & Bp & \forall p \in [1,k-1]\\
\multicolumn{3}{l}{}\\
\multicolumn{3}{l}{}\\
\end{array}
$
&
$
\begin{array}{c@{\rightarrow}cl}
0x & 0x & \forall x \in \{A,B,0,1\}\\
0p & (p+1)(k-1) & \forall p \in [2,k-2]\\
0(k-1) & B(k-1) \\
\multicolumn{3}{l}{}\\
1x & 1A & \forall x \in \{A,B,0\}\\
1p & (p+1)A & \forall p \in [1,k-2]\\ 
1(k-1) & BA \\
\multicolumn{3}{l}{}\\
\multicolumn{3}{l}{\forall p \in [2,k-1]}\\
px & px & \forall x \in \{A,B\} \cup [0,p-1]\\
pn & (n+1)(p-1) & \forall n \in [p,k-2]\\
p(k-1) & B(p-1) \\
\end{array}
$
\\
\\
\end{tabular}
}}

Note that in the initial configuration the sum of the weights of all
agents is exactly $\Sigma a_i x_i \mod k$. Note that the application of
any rule of our protocol does not modify this value, i.e., at any step
of the execution, the sum of the weights of all agents is exactly
$\Sigma a_i x_i \mod k$.

Note that the stable configurations, (i.e., the configurations where
no rule can be applied to modify the state of some agent), are the
following:
\begin{itemize}
\item all agents are in state $0$.
\item there exists a unique agent in state $p \in
  [1,k-2]$ and all other agents are in state $A$.
\item there exists a unique agent in state $p \in
  [2,k-1]$ and all other agents are in state $B$.
\end{itemize}

In any stable configuration, either all agents output $1$ (if $\Sigma
a_ix_i \equiv b \mod k$), or all agents output $0$. We now show that from
any reachable configuration, we can reach a stable configuration.

As in the previous case, as long as there are two agents in states
$p,p' \in [1,k-1]$ with $p \leq p'$, we can apply iteratively the rule
$pp' \rightarrow (p'+1)(p-1)$ between these two agents. Doing so,
either one agent enters state $1$, or one agent is in state $k-1$. If
one agent is in state $1$ while the other is in state $k-1$, we apply
the rule $1(k-1)\rightarrow BA$, and we have decreased the number of
agents with a positive weight by $2$. If there is one agent in state
$1$ (resp. $k-1$) while the other is in state $p \in [1,k-2]$
(resp. $p \in [2,k-1])$, we apply the rule $1p \rightarrow (p+1)A$
(resp. $p(k-1) \rightarrow B(p-1)$) to decrease the number of agents
with a positive weight.  Thus, we can always reach a configuration
where there is at most one agent $a$ in state $p \in [1,k-1]$, while
all the other agents are in state $0$, $A$ or $B$.

If $\Sigma a_ix_i \equiv 0 \mod k$, either all agents have started in
state $0$ or not. In the first case, the initial configuration was a
stable configuration, and we have nothing to prove. In the second
case, the last step before we reach a configuration containing only
agents in state $0$, $A$ or $B$, there was one agent in state $1$, one
agent in state $(k-1)$ and the rule $1(k-1) \rightarrow BA$ has been
applied. Thus, either there exists an agent in state $0$ in the
configuration, or there exists an agent in state $A$ and an agent in
state $B$. In the latter case, we can apply the rule $AB \rightarrow
00$ to be in a configuration containing agents in state $0$. Then,
applying the rules $A0 \rightarrow 00$ and $B0 \rightarrow 00$, we
reach a stable configuration where all agents are in state $0$. 

If $\Sigma a_ix_i \equiv 1 \mod k$ (resp. $\Sigma a_ix_i \equiv k-1
\mod k$), we reach a configuration where there is exactly one agent in
state $1$ (resp. $k-1$) while all the other agents are in state $0$,
$A$ or $B$. Then applying the rules $1B \rightarrow 1A$ and $10
\rightarrow 1B$ (resp. $A(k-1) \rightarrow B(k-1)$ and $0(k-1)
\rightarrow B(k-1)$), we reach a stable configuration where there is
exactly one agent in state $1$ (resp. $k-1$) while all the other
agents are in state $A$ (resp. $B$).

If $\Sigma a_ix_i \equiv p \mod k$ with $p \in [2,k-2]$, we reach a
configuration where there is exactly one agent in state $p$ while all
the other agents are in state $0$, $A$ or $B$.  If the agents with
weight $0$ are either all in state $A$, or all in state $B$, we are in
a stable configuration. Otherwise, applying as long as possible the
rules $AB \rightarrow 00$, $0p \rightarrow (p+1)(k-1)$ and
$(p+1)(k-1)\rightarrow Bp$, we reach a final configuration where
exactly one agent is in state $p$ while all the other agents are in
state $B$.

\end{document}